\documentclass[11pt,oneside]{amsart}
\usepackage[T1]{fontenc}
\usepackage[latin9]{inputenc}
\usepackage{geometry}
\usepackage{amsthm}
\usepackage{amstext}
\usepackage{amssymb}
\usepackage{esint}

\makeatletter

\providecommand{\tabularnewline}{\\}

\numberwithin{equation}{section}
\numberwithin{figure}{section}
\theoremstyle{plain}
\newtheorem{thm}{\protect\theoremname}
  \theoremstyle{definition}
  \newtheorem{defn}[thm]{\protect\definitionname}
  \theoremstyle{plain}
  \newtheorem{lem}[thm]{\protect\lemmaname}
  \theoremstyle{plain}
  \newtheorem{cor}[thm]{\protect\corollaryname}

\usepackage{tikz}
\usepackage{pgf-umlsd}
\usepgflibrary{arrows} % for pgf-umlsd

\makeatother

  \providecommand{\corollaryname}{Corollary}
  \providecommand{\definitionname}{Definition}
  \providecommand{\lemmaname}{Lemma}
\providecommand{\theoremname}{Theorem}

\begin{document}
\global\long\def\E{\mathbb{E}}

\global\long\def\cut#1#2{\text{Cut}_{#1}\left(#2\right)}

\global\long\def\R{\mathbb{R}}

\global\long\def\norm#1{\left\Vert #1\right\Vert }

\global\long\def\defeq{\overset{\text{def}}{=}}

\global\long\def\Ott{\tilde{\tilde{O}}}

\global\long\def\Ot{\tilde{O}}

\global\long\def\OPT{OPT}

\title{Probabilistic Spectral Sparsification In Sublinear Time}

\maketitle
\begin{center}
Yin Tat Lee
\par\end{center}

\begin{center}
yintat@mit.edu
\par\end{center}

\begin{center}
MIT
\par\end{center}
\begin{abstract}
In this paper, we introduce a variant of spectral sparsification,
called probabilistic $(\varepsilon,\delta)$-spectral sparsification.
Roughly speaking, it preserves the cut value of any cut $(S,S^{c})$
with an $1\pm\varepsilon$ multiplicative error and a $\delta\left|S\right|$
additive error. We show how to produce a probabilistic $(\varepsilon,\delta)$-spectral
sparsifier with $O(n\log n/\varepsilon^{2})$ edges in time $\tilde{O}(n/\varepsilon^{2}\delta)$
time for unweighted undirected graph. This gives fastest known sub-linear
time algorithms for different cut problems on unweighted undirected
graph such as 
\begin{itemize}
\item An $\Ot(n/\OPT+n^{3/2+t})$ time $O(\sqrt{\log n/t})$-approximation
algorithm for the sparsest cut problem and the balanced separator
problem.{\small \par}
\item A $n^{1+o(1)}/\varepsilon^{4}$ time approximation minimum s-t cut
algorithm with an $\varepsilon n$ additive error.{\small \par}
\end{itemize}
\end{abstract}

\section{Introduction}

Many cut-based graph problems can be solved approximately in time
$m^{1+o(1)}$, such as the sparsest cut problem, the balanced separator
problem, the minimum s-t cut problem. For dense graphs, we can approximate
graphs by sparse graphs and obtain $O(m)+n^{1+o(1)}$ time approximation
algorithms for different cut-based problems. Unfortunately, in the
era of big data, many dense graphs are too large to process explicitly,
such as distance matrices in machine learning. It is natural to ask
whether it is possible to approximately solve cut-based graph problems
on these graphs in sublinear time.

\subsection{Previous results on sublinear time algorithm for optimization problems}

There are many results on estimating the optimum value of various
combinatorial problems in sublinear time, such as maximum matching
\cite{nguyen2008constant,yoshida2009improved}, minimum vertex cover
\cite{marko2006distance,parnas2007approximating,yoshida2009improved}
and minimum set cover \cite{nguyen2008constant,yoshida2009improved}.
Many of these algorithms simulate \cite{nguyen2008constant} classical
approximation algorithms using local information and transform the
classical algorithms into constant-time algorithms. The running time
of these constant-time algorithms usually depends exponentially on
the maximum degree of the graph and the additive error $\delta$.
Unfortunately, there has been little progress for dense graphs because
of the limitation of this simulation approach. The only result for
dense graphs we aware of is an $\Ot(n\cdot poly(1/\varepsilon))$
time algorithm for finding an factor-2 approximation of the size of
a maximum vertex cover within an extra $\varepsilon n$ additive error
\cite{onak2012near}.

Instead of using the simulation approach, we suggest another principled
way to obtain sublinear time algorithms - sparsification.

\subsection{Sparsification}

In this work, we heavily use the concept of sparsification from the
spectral graph theory. Benczúr and Karger \cite{benczur1996approximating}
introduced the notation of cut sparsification for solving cut-based
problem on dense graphs, but it is not designed for sublinear time
algorithms. A graph $H$ is called a cut sparsifier of $G=(V,E,\omega)$
if $H$ is a sparse graph on $V$ such that the cut value of any cut
in $H$ is within a factor of $(1\pm\varepsilon)$ of its value in
$G$. In other words, for all characteristic vectors $x\in\{0,1\}^{V}$,
we have 
\begin{equation}
\sum_{u\sim v}\tilde{\omega}_{uv}(x(u)-x(v))^{2}\in(1\pm\varepsilon)\sum_{u\sim v}\omega_{uv}(x(u)-x(v))^{2}\label{eq:sparsifier_req}
\end{equation}
where $\omega$ and $\tilde{\omega}$ are the weights of edges in
graph $G$ and $H$ respectively. They proved that sampling a graph
with certain probability gives a cut sparsifier and the sampling probability
can be computed in time $\Ot(m)$. This gives an $\Ot(m)$ time algorithm
to find a cut sparsifier with $\Ot(n/\varepsilon^{2})$ edges. \cite{karger1997using,karger1998better,karger1998finding}
used the cut sparsification to obtain various fast algorithms for
the minimum s-t cut problem and the maximum flow problem for dense
graphs. Besides this, cut sparsification has many other applications
because of its strong guarantee. Of particular relevance to this paper,
M\k{a}dry \cite{madry2010fast} used the cut sparsification as one
of the essential components to give a way to reduce cut problems on
general graph to some almost trees and obtained almost linear time
algorithms for many cut problems.

Inspired by the cut sparsification, Spielman and Teng \cite{spielman2011spectral}
defined the notation of spectral sparsification, which is a stronger
notation of sparsification. It requires the graph $H$ satisfies (\ref{eq:sparsifier_req})
for all vectors $x\in\mathbb{R}^{V}$. From numerical perspective,
it is same as requiring the Laplacian of the graph $H$ is a good
preconditioner of the Laplacian of the graph $G$. So, many equations
related to the Laplacian $G$, such as, Laplacian equation, eigenvalue
problem, heat equation, random walk, can be solved in the graph $H$
within a certain error. Spielman and Srivastava \cite{spielman2011graph}
showed that spectral sparsifiers can be found by sampling the graph
with probability proportional to effective resistances. And they presented
an algorithm to estimate effective resistances in time $\Ot(m)$ using
nearly linear time Laplacian solver \cite{koutis2011nearly,kelner2013simple,lee2013efficient}.

Although there are a lot of results for the streaming model \cite{goel2010graph,kelner2011spectral,goel2012single},
there is no sublinear time algorithm because it is apparently impossible.

\subsection{Our contribution}

Motivated by the sublinear time problem and the spectral graph theory,
we introduce a variant of spectral sparsification \cite{spielman2011spectral}
that we call probabilistic spectral sparsification. Given an unweighted
graph $G=(V,E)$, a probabilistic $(\varepsilon,\delta)$-spectral
sparsifier of the graph $G$ is a weighted random graph $\tilde{G}=(V,\tilde{E},\tilde{\omega})$
on the vertex set $V$ such that
\begin{enumerate}
\item Lower Bound: We have 
\begin{equation}
(1-\varepsilon)\sum_{(x,y)\in E}\left(u(x)-u(y)\right)^{2}\leq\sum_{(x,y)\in\tilde{E}}\tilde{\omega}(x,y)\left(u(x)-u(y)\right)^{2}\quad\text{for all }u\in\mathbb{R}^{V}.\label{eq:lower_bound}
\end{equation}

\item Upper Bound %
\footnote{In this paper, high probability means a constant probability sufficiently
close to $1$.%
}: For all $u\in\mathbb{R}^{V}$, we have
\begin{equation}
\sum_{(x,y)\in\tilde{E}}\tilde{\omega}(x,y)\left(u(x)-u(y)\right)^{2}\leq(1+\varepsilon)\sum_{(x,y)\in E}\left(u(x)-u(y)\right)^{2}+\delta\norm u_{2}^{2}\quad\text{with high probability}.\label{eq:upper_bound}
\end{equation}

\end{enumerate}
It seems to us that standard matrix concentration bound can at best
give bounds like $\delta\sum d(x)u^{2}(x)$ and there are results
\cite{frieze1999quick,frieze2000edge} on this line concerning fast
approximate general matrix without paying $\tilde{O}\left(m\right)$
time to compute effective resistances. However, the guarantee $\delta\sum d(x)u^{2}(x)$
can be $n$ times worse than $\delta\norm u_{2}^{2}$ for dense matrices
and it is not good enough for certain applications such as the sparsest
cut problem.

In this paper, we show how to construct a probabilistic $(\varepsilon,\delta)$-spectral
sparsifier with $\Ot(n/\varepsilon^{2})$ edges in time $\Ot\left(n/\varepsilon^{2}\delta\right)$.
We avoid the matrix concentration bound by using graph structures
and obtain this almost tight result. As a result, this transforms
many cut problems on dense graphs into sparse graphs and hence gives
sublinear algorithms on a bunch of cut-based problems. We illustrate
the applicability of our sparsification on the following fundamental
cut-based graph problems
\begin{itemize}
\item An $\Ot(n/\OPT+n^{3/2+t})$ time $O(\sqrt{\log n/t})$-approximation
algorithm for the sparsest cut problem and the balanced separator
problem.
\item An $\Ot\left(n/\OPT+2^{k}n^{1+1/(3\cdot2^{k}-1)+o(1)}\right)$ time
$O\left(\log^{\left(1+o(1)\right)\left(k+1/2\right)}n\right)$-approximation
algorithm for the sparsest cut problem and the balanced separator
problem.
\item An $\Ot(\sqrt{mn}/\varepsilon^{3})$ time and a $n^{1+o(1)}/\varepsilon^{4}$
time approximation minimum s-t cut algorithm with an $\varepsilon n$
additive error.
\end{itemize}
This sparsifier is a weaker notion than the spectral sparsification
introduced by Spielman and Teng \cite{spielman2011spectral}, which
requires a single graph to satisfy both upper and lower bounds with
$\delta=0$. To justify our notion, we show that it takes at least
$\Omega\left(n/\varepsilon^{2}+n/\delta\right)$ time to find this
sparsifier and hence the extra additive term is unavoidable. Furthermore,
we show in Theorem \ref{thm:lowerbound-1} that the term $n/\OPT$
in the running time shown above is unavoidable for the sparsest cut
problem.

\subsection{Definitions}

Let $[n]=\{1,2,\cdots n\}$. The notation $\Ot(f(n))$ means $O(f(n)\log^{c}(n))$
for some constant $c$ and $\Ott(f(n))$ means $O(f(n)\log^{c}\log(n))$
for some constant $c$. Let $G$ be a weighted undirected graph with
$n$ vertices and $m$ edges with weights $\omega$. We write $(u,v)\in G$
if the vertex $u$ is adjacent to the vertex $v$ in the graph $G$.
Let the neighborhood of $v$ be $N_{G}(v)\defeq\{u:(u,v)\in G\}$.
Let $d_{G}(u)$ be the weighted degree of the vertex $u$, that is
$d_{G}(u)=\sum_{(u,v)\in G}\omega(u,v)$. The cut value of $U$ is
defined by $\cut GU=\sum_{(u,v)\in G,u\in U,v\notin U}\omega(e)$. 
\begin{defn}
Given a weighted undirected graph $G$, we view the graph $G$ as
an electric network and define the resistance of an edge $(s,t)$
is $1/\omega(s,t)$. The effective resistance $R(s,t)$ is the potential
difference between $s$ and $t$ when there is a unit flow send from
$s$ to $t$ on this electric network.
\end{defn}

\begin{defn}
(General Graph Model) In the general graph model, a graph $G=(V,E)$
is represented by the number of vertices $n$ and three oracles 
\begin{enumerate}
\item The vertex oracle $\mathcal{O}_{1}:[n]\rightarrow V$ which returns
the $i$-th vertex of the graph.
\item The degree oracle $\mathcal{O}_{2}:V\rightarrow\mathbb{Z}^{+}$ which
returns the degree $d(v)$.
\item The edge oracle $\mathcal{O}_{3}:V\times\mathbb{Z}^{+}\rightarrow V$
which returns the $i$-th vertex adjacent to $v$.
\end{enumerate}
\end{defn}

\section{Probabilistic Spectral Sparsification}

In this section, we show how to construct probabilistic spectral sparsifiers
in sublinear time. The algorithm is inspired by the following two
results about effective resistance. Spielman and Srivastava \cite{spielman2011graph}
shows that sampling edges proportional to the effective resistances
of edges produce a spectral sparsifier. It is known that on an unweighted
expander, we have \cite{lovasz1993random}
\begin{equation}
R(s,t)=\Theta\left(\frac{1}{d(s)}+\frac{1}{d(t)}\right)\label{eq:eff_est}
\end{equation}
for any edge $(s,t)$. These two results show that we can construct
spectral sparsifiers for expanders according to the degree of vertices.
Therefore, if we can transform a graph into an expander by modifying
only some edges, then we can obtain a spectral sparsifier with small
additive error. Unfortunately, it requires modifying $O(m)$ edges
which is too large for certain problems. Instead of satisfying the
expander condition for (\ref{eq:eff_est}), we show how to make a
graph satisfies (\ref{eq:eff_est}) directly by adding only a few
edges. To do this, we randomly select a subset of the graph and put
a sparse expander on this subset. In Lemma \ref{lem:graph_smoothing},
we show that the effective resistances in this new graph satisfies
the estimate (\ref{eq:eff_est}). This gives an algorithm to construct
probabilistic spectral sparsifiers.

In this paper, the only property of expander we used is that there
are lots of edge-disjoint short paths in an expander.
\begin{thm}
\label{thm: expander}\cite{lubotzky1988ramanujan,frieze2000edge}There
is an $O(n)$ time algorithm to construct a graph $E_{n}$ such that
\begin{enumerate}
\item It has $\Theta(n)$ vertices, $O(n)$ edges and the maximum degree
is $\Theta(1)$.
\item For any pairs $\{(a_{i},b_{i})\}_{i=1}^{k}$ with $k=O(\frac{n}{\log n})$,
there exists edge-disjoint paths of length $O(\log n)$ in $E_{n}$
joining $a_{i}$ to $b_{i}$.
\end{enumerate}
\end{thm}
The following key lemma shows that putting $E_{n}$ in a random subset
of $G$ makes the graph satisfies (\ref{eq:eff_est}).
\begin{lem}
\label{lem:graph_smoothing}Assume $\delta\leq1/\log n$. Given an
unweighted undirected graph $G=(V,E)$. Let $E_{\delta n}$ be the
graph given by Theorem \ref{thm: expander} and $V_{\delta}$ be a
random subset of $V$ with size $\left|E_{\delta n}\right|$. We view
$E_{\delta n}$ as a graph on $V_{\delta}$ and let $\tilde{G}$ be
the union of $G$ and $E_{\delta n}$. With high probability, for
any edge $(s,t)$, we have
\[
\frac{1}{2}\left(\frac{1}{d_{\tilde{G}}(s)}+\frac{1}{d_{\tilde{G}}(t)}\right)\leq R_{\tilde{G}}(s,t)\leq O\left(\frac{\log n}{\delta}\left(\frac{1}{d_{\tilde{G}}(s)}+\frac{1}{d_{\tilde{G}}(t)}\right)\right).
\]
\end{lem}
\begin{proof}
Claim: With high probability, for any vertex $v$ with $d_{\tilde{G}}(v)=\Omega\left(\log n/\delta\right)$,
\[
\left|V_{\delta}\cap N_{G}(v)\right|=\Omega\left(\delta d_{\tilde{G}}(v)\right).
\]
Assume the claim. Let $(s,t)$ be any edge. Write $d_{\tilde{G}}(s)$
as $d(s)$ and $d_{\tilde{G}(t)}$ as $d(t)$ for simplicity. Since
the effective resistance of an edge is bounded by 1 for unweighted
graph, if $d(s)$ or $d(t)$ is at most $O\left(\log n/\delta\right)$,
we have
\begin{eqnarray*}
R_{\tilde{G}}(s,t) & \leq & 1=O\left(\frac{\log n}{\delta}\left(\frac{1}{d(s)}+\frac{1}{d(t)}\right)\right).
\end{eqnarray*}
Hence, we can assume both $d(s)$ and $d(t)$ is at least $\Omega(\log n/\delta)$.
The claim shows that there are at least $\Omega\left(\delta d(s)\right)$
vertices of $V_{\delta}$ in the neighbor $N_{G}(s)$ of $s$ and
at least $\Omega\left(\delta d(t)\right)$ for $t$. Since $\delta d(s)\leq n/\log n$,
Theorem \ref{thm: expander} shows that there are $\Omega\left(\delta\min\left(d(s),d(t)\right)\right)$
edge-disjoint paths with length $O(\log n)$ joining these neighbor
of $s$ to these neighbor of $t$. By Rayleigh's Monotonicity Principle,
the effective resistance between $s$ and $t$ is less than the graph
with only $\Omega\left(\delta\min\left(d(s),d(t)\right)\right)$ edge-disjoint
paths from $s$ to $t$ with length $O(\log n)$. Hence, we have
\[
R_{\tilde{G}}(s,t)=O\left(\frac{\log n}{\delta\min\left(d(s),d(t)\right)}\right)=O\left(\frac{\log n}{\delta}\left(\frac{1}{d(s)}+\frac{1}{d(t)}\right)\right).
\]
Therefore, in both case, we have
\[
R_{\tilde{G}}(s,t)\leq O\left(\frac{\log n}{\delta}\left(\frac{1}{d(s)}+\frac{1}{d(t)}\right)\right).
\]
Another side of the inequality comes from \cite{lovasz1993random}.

Proof of the claim: Let $U$ be any subset of $V$ with $k$ elements.
Note that $X=\left|V_{\delta}\cap U\right|$ is a random variable
with hypergeometric distribution. The Chernoff bound for hypergeometric
distribution \cite[Thm 1.17]{doerr2011analyzing} shows that $\mathbb{P}\left(X\leq\frac{1}{2}\mathbb{E}(X)\right)\leq\left(\frac{2}{e}\right)^{\mathbb{E}(X)/2}.$
For $k=\Omega\left(\log n/\delta\right)$, we have $\mathbb{E}(X)=\delta k=\Omega\left(\log n\right)$
and hence $\mathbb{P}\left(X\leq\frac{\delta k}{2}\right)=\frac{1}{\text{poly}(n)}$.
Since there are only $n$ neighbor sets $N_{G}(v)$, union bound shows
that with high probability, for any $v\in V$ with $d_{G}(v)=\Omega\left(\log n/\delta\right)$,
we have 
\begin{eqnarray*}
\left|V_{\delta}\cap N_{G}(v)\right| & = & \Omega\left(\delta d_{G}(v)\right)=\Omega\left(\delta d_{\tilde{G}}(v)\right)
\end{eqnarray*}
where the last line comes from the fact that the maximum degree of
$E_{\delta n}$ is $O(1)$. 
\end{proof}
Having a good estimate of effective resistances, we could use the
following algorithm proposed by Spielman and Srivastava \cite{spielman2011graph}
to construct a spectral sparsifier of $\tilde{G}$.

\begin{center}
\begin{tabular}{|l|}
\hline 
$H=\mathbf{Sparsify}(G,p,q)$\tabularnewline
\hline 
\hline 
1. Repeat $q$ times:\tabularnewline
\hline 
1a. Sample an edge $e$ from $G$ with probability $p(e)$.\tabularnewline
\hline 
1b. Add it to $H$ with weight $\left(qp(e)\right)^{-1}$.\tabularnewline
\hline 
\end{tabular}
\par\end{center}
\begin{thm}
\label{thm: mul_spasifier}\cite{spielman2011graph} Let $G$ be an
unweighted undirected graph. Suppose $p(e)$ are numbers such that
$\sum p(e)=1$ and
\[
p(e)\geq\frac{R(e)}{\alpha n}
\]
for some $\alpha>0$. Then, with high probability, $\mathbf{Sparsify}\left(G,p,\Theta\left(\alpha n\log n/\varepsilon^{2}\right)\right)$
is a $\varepsilon$-spectral sparsifier with $O(\alpha n\log n/\varepsilon^{2})$
edges in time $O(\alpha n\log n/\varepsilon^{2})$.
\end{thm}
Since the algorithm $\mathbf{Sparsify}$ cannot provide the optimal
sparsity when $\alpha\gg1$, we will use the spectral sparsification
algorithm proposed by Koutis, Levin and Peng \cite{koutis2012faster}
to further sparsify the graph at the end.
\begin{thm}
\label{thm: fast_spasifier}\cite{koutis2012faster} There is a spectral
sparsification algorithm, we call $\mathbf{FastSparsify}\left(G\right)$,
that produces a $\varepsilon-$spectral sparsifier with $O(n\log n/\varepsilon^{2})$
edges in time $\Ott(m\log^{2}n\log(1/\varepsilon))$ with high probability.
\end{thm}
Using Lemma \ref{lem:graph_smoothing}, Theorem \ref{thm: mul_spasifier}
and Theorem \ref{thm: fast_spasifier}, we can derive our main theorem:

\begin{center}
\begin{tabular}{|l|}
\hline 
$\overline{H}=\mathbf{SublinearSparsify}(G,\varepsilon,\delta)$\tabularnewline
\hline 
\hline 
1. Let $E_{\delta n}$ be the graph given by Theorem \ref{thm: expander}.\tabularnewline
\hline 
2. Let $V_{\delta}$ be a random subset of $V$ with size $\left|E_{\delta n}\right|$.\tabularnewline
\hline 
3. View $E_{\delta n}$ as a graph on $V_{\delta}$ and let $\tilde{G}$
be the union of $G$ and $E_{\delta n}$\tabularnewline
\hline 
4. Let $p(u,v)=1/\left(nd_{\tilde{G}}(u)\right)+1/\left(nd_{\tilde{G}}(v)\right).$\tabularnewline
\hline 
5. $H=\mathbf{Sparsify}\left(\tilde{G},p,\Theta(n\log^{2}n/\delta\varepsilon^{2})\right).$\tabularnewline
\hline 
6. $\overline{H}=\mathbf{FastSparsify}\left(H\right)$.\tabularnewline
\hline 
\end{tabular}
\par\end{center}
\begin{thm}
\label{thm:sparsify_algo}Assume $\delta\leq1/\log n$ and $\varepsilon<1$
and the General Graph Model. With high probability, the $\mathbf{SublinearSparsify}(G,\varepsilon,\delta)$
algorithm produces a probabilistic $(O(\varepsilon),O(\delta))$-spectral
sparsifier with $O(n\log n/\varepsilon^{2})$ edges in time $\Ott(n\log^{4}n\log(1/\varepsilon)/\delta\varepsilon^{2}).$%
\footnote{$\Ott(f(n))$ means $O(f(n)\log^{c}\log(n))$ for some constant $c$.%
}\end{thm}
\begin{proof}
Lemma \ref{lem:graph_smoothing} shows that with high probability,
for all $(u,v)$, we have
\begin{eqnarray*}
p(u,v) & = & \frac{1}{n}\left(\frac{1}{d_{\tilde{G}}(u)}+\frac{1}{d_{\tilde{G}}(v)}\right)=\Omega\left(\frac{\delta}{\log n}\right)\frac{R_{\tilde{G}}(u,v)}{n}.
\end{eqnarray*}
Hence, $p$ satisfy the assumption of Theorem \ref{thm: mul_spasifier}
with $\alpha=O\left(\log n/\delta\right)$. Therefore, $H$ is a $\varepsilon$-spectral
sparsifier of $\tilde{G}$ with high probability. For any $u\in\mathbb{R}^{V}$,
we have
\begin{eqnarray*}
\sum_{(x,y)\in H}\omega(x,y)\left(u(x)-u(y)\right)^{2} & \geq & (1-\varepsilon)\sum_{(x,y)\in\tilde{G}}\left(u(x)-u(y)\right)^{2}\\
 & \geq & (1-\varepsilon)\sum_{(x,y)\in G}\left(u(x)-u(y)\right)^{2}.
\end{eqnarray*}
Hence, $H$ satisfies the condition (\ref{eq:lower_bound}). Also,
for any $u\in\mathbb{R}^{V}$, we have
\begin{eqnarray*}
\sum_{(x,y)\in H}\omega(x,y)\left(u(x)-u(y)\right)^{2} & \leq & (1+\varepsilon)\sum_{(x,y)\in\tilde{G}}\left(u(x)-u(y)\right)^{2}\\
 & \leq & (1+\varepsilon)\sum_{(x,y)\in G}\left(u(x)-u(y)\right)^{2}+4\sum_{x\in V_{\delta}}\left(u(x)\right)^{2}.
\end{eqnarray*}
Since $V_{\delta}$ is a random subset of $V$ with size $\Theta(\delta n)$,
we have
\[
\mathbb{E}\left(\sum_{x\in V_{\delta}}\left(u(x)\right)^{2}\right)=\Theta\left(\delta\right)\sum_{x\in V}\left(u(x)\right)^{2}.
\]
Thus, for any $u\in\mathbb{R}^{V}$, with high probability,
\[
\sum_{(x,y)\in H}\omega(x,y)\left(u(x)-u(y)\right)^{2}\leq(1+\varepsilon)\sum_{(x,y)\in G}\left(u(x)-u(y)\right)^{2}+\Theta\left(\delta\right)\norm u^{2}.
\]
Hence, $H$ satisfies the condition (\ref{eq:upper_bound}). Therefore,
$H$ is a probabilistic $(O(\varepsilon),O(\delta))$-spectral sparsifier
with $O(n\log^{2}n/\delta\varepsilon^{2})$ edges. Using Theorem \ref{thm: fast_spasifier}
and similar proof, we obtain that $\overline{H}$ is a probabilistic
$(O(\varepsilon),O(\delta))$-spectral sparsifier with $O(n\log n/\varepsilon^{2})$
edges.

Since the sampling probability is of the form $1/d(s)+1/d(t)$, we
do it by sampling each node with probability proportionally to the
degree. Thus, it can be implemented in time $O(\log n)$ using the
General Graph Model.
\end{proof}

\section{Applications}

In this section, we demonstrate how to apply the probabilistic spectral
sparsification to solve cut-based problems. Restricting our focus
on $x\in\{0,1\}^{V}$, the upper bound (\ref{eq:upper_bound}) and
the lower bound (\ref{eq:lower_bound}) of the probabilistic spectral
sparsification becomes the following: Suppose $\tilde{G}$ is a probabilistic
$(\varepsilon,\delta)$-spectral sparsifier of $G$, then we have
\begin{enumerate}
\item Lower Bound: We have 
\begin{equation}
(1-\varepsilon)\cut GU\leq\cut{\tilde{G}}U\quad\text{for all }U\subset V.\label{eq:lower_bound-1}
\end{equation}

\item Upper Bound: For all $U\subset V$, we have
\begin{equation}
\cut{\tilde{G}}U\leq(1+\varepsilon)\cut GU+\delta\left|U\right|\quad\text{with high probability}.\label{eq:upper_bound-1}
\end{equation}

\end{enumerate}
The lower bound shows that any cut with a small cut value in $\tilde{G}$
has a small cut value in $G$ and the upper bound shows that such
cut with a small cut value exists in $\tilde{G}$ with high probability.
Therefore, as long as the additive error $\delta\left|U\right|$ is
acceptable, we can approximately solve any cut-based problem on a
probabilistic spectral sparsifier of the original graph and use the
upper bound and lower bound to certify that it is a good solution
for the original graph.

\subsection{(Uniform) Sparsest Cut Problem and Balanced Separator Problem}

The sparsest cut problem is to find a set $U$ with $\left|U\right|<n/2$
such that it minimizes the ratio of $\cut GU$ and $\left|U\right|$.
The balanced separator problem is to solve the same problem with an
extra condition $\left|U\right|=\Omega(n)$. The best known algorithm
\cite{arora2009expander} for both problems achieves an $O(\sqrt{\log n})$
approximation ratio in polynomial time. For fast algorithms, Sherman
\cite{sherman2009breaking} gives an $\Ot\left(m+n^{3/2+t}\right)$
time algorithm with approximation ratio $O\left(\sqrt{\log n/t}\right)$
and M\k{a}dry \cite{madry2010fast} gives an $\tilde{O}\left(m+2^{k}n^{1+1/(3\cdot2^{k}-1)+o(1)}\right)$
time algorithm with approximation ratio $O\left(\log^{\left(1+o(1)\right)\left(k+1/2\right)}n\right)$
for all $k\geq1$. Both algorithms works for weighted graph. Using
these results and our probabilistic spectral sparsifiers, we have
the following:
\begin{cor}
Assume the graph is undirected and unweighted. For any $t\in[O(1/\log n),\Omega(1)]$,
there is an $\tilde{O}\left(n/\OPT+n^{3/2+t}\right)$ time algorithm
to approximate the sparsest cut problem and the balanced separator
problem with approximation ratio $O\left(\sqrt{\log n/t}\right)$.
For any integer $k\geq1$, there is an $\tilde{O}\left(n/\OPT+2^{k}n^{1+1/(3\cdot2^{k}-1)+o(1)}\right)$
time algorithm with approximation ratio $O\left(\log^{\left(1+o(1)\right)\left(k+1/2\right)}n\right)$.\end{cor}
\begin{proof}
The proof for both problems and both approximation ratios are similar.
Assume it is the sparsest cut problem and we want to get an $\alpha$
approximation algorithm. The algorithm works as follows:
\begin{enumerate}
\item Take $\delta=1/\log n$.
\item Let $\tilde{G}$ be a probabilistic $(\frac{1}{2},\delta)$-spectral
sparsifier of $G$.
\item Find an $\alpha$ approximate sparsest cut $\overline{U}$ on the
graph $\tilde{G}$.
\item Let $\overline{\OPT}$ be the ratio of $\cut{\tilde{G}}{\overline{U}}$
and $\left|\overline{U}\right|$.
\item If $\delta>\overline{\OPT}/2\alpha$

\begin{enumerate}
\item $\delta\leftarrow\delta/2$, go to step 2
\item Otherwise, output $\overline{U}$.
\end{enumerate}
\end{enumerate}
Let $G$ be the original graph. Let $U_{G}$ and $\OPT_{G}$ are an
optimum set and the optimum value for this problem in graph $G$.
Let $\OPT_{\tilde{G}}$ is the optimum value for graph $\tilde{G}$.
Using (\ref{eq:upper_bound-1}), we have
\begin{eqnarray*}
\OPT_{\tilde{G}} & \leq & \frac{\cut{\tilde{G}}{U_{G}}}{\left|U_{G}\right|}\leq\frac{\frac{3}{2}\cut G{U_{G}}+\delta\left|U_{G}\right|}{\left|U_{G}\right|}=\frac{3}{2}\OPT_{G}+\delta.
\end{eqnarray*}
Since $\overline{U}$ is an $\alpha$ approximate sparsest cut on
$\tilde{G}$, we have
\[
\frac{1}{\alpha}\overline{\OPT}\leq OPT_{\tilde{G}}\leq\frac{3}{2}\OPT_{G}+\delta.
\]
If $\delta<\overline{\OPT}/2\alpha$, then we have $\overline{\OPT}\leq3\alpha\OPT_{G}$.
Hence, (\ref{eq:lower_bound-1}) gives $\cut G{\overline{U}}/\left|\overline{U}\right|\leq6\alpha\OPT_{G}$
and the set $\overline{U}$ solve the problem in $G$ with approximation
ratio $6\alpha$. Otherwise, we have $\delta$ decrease by $2$. Since
$\overline{\OPT}\geq\frac{1}{n}$, the algorithm takes at most $\log n$
iterations. 
\end{proof}
In Theorem \ref{thm:lowerbound-1}, we show that the term $n/\OPT$
in running time is unavoidable. So, our reduction is almost optimal.

\subsection{Minimum s-t Cut Problem}

The Minimum s-t Cut Problem is to find a set $U$ such that $s\in U$,
$t\notin U$ and it minimizes $\cut GU$. 
\begin{cor}
Assume the graph is undirected and unweighted. There are an $\tilde{O}(\sqrt{mn}/\varepsilon^{3})$
time and a $n^{1+o(1)}/\varepsilon^{4}$ time algorithm to find a
minimum s-t cut up to an $\varepsilon n$ error.\end{cor}
\begin{proof}
On an undirected graph with integer weight, the proof of Theorem 4
of \cite{lee2013new} shows an $\tilde{O}\left(\frac{m}{\varepsilon}\sqrt{\frac{W}{n}}\right)$
time algorithm to compute an approximate minimum s-t cut with $\varepsilon n$
additive error, where $W$ is the total weight. Note that the total
weight of the result of our sparsification is $\tilde{O}\left(m\right)$
and changes can be made so that the weights are integers. This gives
the first result.

For the second result, it follows from \cite{kelner2013almost,sherman2013nearly}.
\end{proof}

\subsection{Other applications}

For some cut-based problems such as the maximum cut problem and the
minimum cut problem, sampling edges with constant probability gives
good enough guarantee. For other cut-based problems such as the multicut
problem, one can use our sparsification to reduce the problem into
sparse graphs, then use the technique by M\k{a}dry \cite{madry2010fast}
to further reduce the problem into almost trees, which can be then
solved by elementary methods in many cases.

Our probabilistic spectral sparsifier is also useful for applications
involves the graph energy $\sum_{x\sim y}(u(x)-u(y))^{2}$. It includes
a lot of problems in many fields, such as approximating Fiedler vector
\cite{levy2006laplace}, minimizing all sorts of variational problems
in image processing \cite{chan2005image}.

\section{Lower Bound}

In this section, we show that the additive error in upper bound (\ref{eq:upper_bound})
for the sparsifier is necessary. In the proof, we construct a family
of random graphs and shows that it is difficult to estimate the cut
value of some sets in the graphs. In the Lemma \ref{lem:Gkp}, we
construct a family of random graphs which is served as a building
block of the graphs for Theorem \ref{thm:lowerbound}. 
\begin{lem}
\label{lem:Gkp}Assume the general graph model. For any integer $k>3$
and $0<p\leq1/4$ such that $pk^{2}\geq100$, there is a family of
random graph $G_{k,p}=(V,E)$ with $4k$ vertices and $2k^{2}$ edges
and a cut $S\subset E$ which satisfies the following property: let
$C$ be the estimate of $\cut{}S$ of any deterministic algorithm
which calls the oracle less than $k^{2}/2$ times. Then, we have 
\[
\mathbb{P}\left(\left|C-\cut{}S\right|\geq k\sqrt{\frac{p}{8}}\right)\geq0.01.
\]
\end{lem}
\begin{proof}
For each pair $i,j\in[k]$, let $H_{ij}$ be an independent variable
such that $H_{ij}=1$ with probability $p$ and $H_{ij}=0$ otherwise.
We construct a family of random graph $G_{k,p}$ using the random
variable $\{H_{ij}\}_{i\in[k],j\in[k]}.$ The graph $G_{k,p}$ consists
of 4 sets of vertices $V^{1},V^{2},V^{3},V^{4}$ and each of them
has $k$ vertices. We call each vertex in $V^{t}$ by $it$ for some
$i\in[k]$. If $H_{ij}=1$, we place the edges $\{(i1,j2),(j3,i4)\}$,
which is indicated by the solid lines in the figure. Otherwise, we
place the edges $\{(i1,j3),(j2,i4)\}$. Note that this graph is $k$
regular and hence the degree oracle does not provide any information.

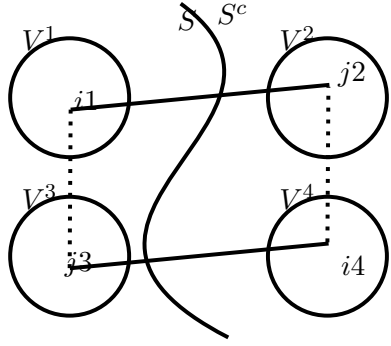
\begin{figure}
\ifx\du\undefined
  \newlength{\du}
\fi
\setlength{\du}{15\unitlength}
\begin{tikzpicture}[scale=0.5]
\pgftransformxscale{1.000000}
\pgftransformyscale{-1.000000}
\definecolor{dialinecolor}{rgb}{0.000000, 0.000000, 0.000000}
\pgfsetstrokecolor{dialinecolor}
\definecolor{dialinecolor}{rgb}{1.000000, 1.000000, 1.000000}
\pgfsetfillcolor{dialinecolor}
\definecolor{dialinecolor}{rgb}{1.000000, 1.000000, 1.000000}
\pgfsetfillcolor{dialinecolor}
\pgfpathellipse{\pgfpoint{11.000000\du}{9.000000\du}}{\pgfpoint{3.000000\du}{0\du}}{\pgfpoint{0\du}{3.000000\du}}
\pgfusepath{fill}
\pgfsetlinewidth{0.100000\du}
\pgfsetdash{}{0pt}
\pgfsetdash{}{0pt}
\pgfsetmiterjoin
\definecolor{dialinecolor}{rgb}{0.000000, 0.000000, 0.000000}
\pgfsetstrokecolor{dialinecolor}
\pgfpathellipse{\pgfpoint{11.000000\du}{9.000000\du}}{\pgfpoint{3.000000\du}{0\du}}{\pgfpoint{0\du}{3.000000\du}}
\pgfusepath{stroke}
% setfont left to latex
\definecolor{dialinecolor}{rgb}{0.000000, 0.000000, 0.000000}
\pgfsetstrokecolor{dialinecolor}
\node at (11.000000\du,9.240000\du){};
% setfont left to latex
\definecolor{dialinecolor}{rgb}{0.000000, 0.000000, 0.000000}
\pgfsetstrokecolor{dialinecolor}
\node[anchor=west] at (8.000000\du,6.000000\du){$V^1$};
\definecolor{dialinecolor}{rgb}{1.000000, 1.000000, 1.000000}
\pgfsetfillcolor{dialinecolor}
\pgfpathellipse{\pgfpoint{24.000000\du}{9.000000\du}}{\pgfpoint{3.000000\du}{0\du}}{\pgfpoint{0\du}{3.000000\du}}
\pgfusepath{fill}
\pgfsetlinewidth{0.100000\du}
\pgfsetdash{}{0pt}
\pgfsetdash{}{0pt}
\pgfsetmiterjoin
\definecolor{dialinecolor}{rgb}{0.000000, 0.000000, 0.000000}
\pgfsetstrokecolor{dialinecolor}
\pgfpathellipse{\pgfpoint{24.000000\du}{9.000000\du}}{\pgfpoint{3.000000\du}{0\du}}{\pgfpoint{0\du}{3.000000\du}}
\pgfusepath{stroke}
% setfont left to latex
\definecolor{dialinecolor}{rgb}{0.000000, 0.000000, 0.000000}
\pgfsetstrokecolor{dialinecolor}
\node at (24.000000\du,9.240000\du){};
% setfont left to latex
\definecolor{dialinecolor}{rgb}{0.000000, 0.000000, 0.000000}
\pgfsetstrokecolor{dialinecolor}
\node[anchor=west] at (21.000000\du,6.000000\du){$V^2$};
\definecolor{dialinecolor}{rgb}{1.000000, 1.000000, 1.000000}
\pgfsetfillcolor{dialinecolor}
\pgfpathellipse{\pgfpoint{11.000000\du}{17.000000\du}}{\pgfpoint{3.000000\du}{0\du}}{\pgfpoint{0\du}{3.000000\du}}
\pgfusepath{fill}
\pgfsetlinewidth{0.100000\du}
\pgfsetdash{}{0pt}
\pgfsetdash{}{0pt}
\pgfsetmiterjoin
\definecolor{dialinecolor}{rgb}{0.000000, 0.000000, 0.000000}
\pgfsetstrokecolor{dialinecolor}
\pgfpathellipse{\pgfpoint{11.000000\du}{17.000000\du}}{\pgfpoint{3.000000\du}{0\du}}{\pgfpoint{0\du}{3.000000\du}}
\pgfusepath{stroke}
% setfont left to latex
\definecolor{dialinecolor}{rgb}{0.000000, 0.000000, 0.000000}
\pgfsetstrokecolor{dialinecolor}
\node at (11.000000\du,17.240000\du){};
% setfont left to latex
\definecolor{dialinecolor}{rgb}{0.000000, 0.000000, 0.000000}
\pgfsetstrokecolor{dialinecolor}
\node[anchor=west] at (8.000000\du,14.000000\du){$V^3$};
\definecolor{dialinecolor}{rgb}{1.000000, 1.000000, 1.000000}
\pgfsetfillcolor{dialinecolor}
\pgfpathellipse{\pgfpoint{24.000000\du}{17.000000\du}}{\pgfpoint{3.000000\du}{0\du}}{\pgfpoint{0\du}{3.000000\du}}
\pgfusepath{fill}
\pgfsetlinewidth{0.100000\du}
\pgfsetdash{}{0pt}
\pgfsetdash{}{0pt}
\pgfsetmiterjoin
\definecolor{dialinecolor}{rgb}{0.000000, 0.000000, 0.000000}
\pgfsetstrokecolor{dialinecolor}
\pgfpathellipse{\pgfpoint{24.000000\du}{17.000000\du}}{\pgfpoint{3.000000\du}{0\du}}{\pgfpoint{0\du}{3.000000\du}}
\pgfusepath{stroke}
% setfont left to latex
\definecolor{dialinecolor}{rgb}{0.000000, 0.000000, 0.000000}
\pgfsetstrokecolor{dialinecolor}
\node at (24.000000\du,17.240000\du){};
% setfont left to latex
\definecolor{dialinecolor}{rgb}{0.000000, 0.000000, 0.000000}
\pgfsetstrokecolor{dialinecolor}
\node[anchor=west] at (21.000000\du,14.000000\du){$V^4$};
\pgfsetlinewidth{0.100000\du}
\pgfsetdash{}{0pt}
\pgfsetdash{}{0pt}
\pgfsetbuttcap
{
\definecolor{dialinecolor}{rgb}{0.000000, 0.000000, 0.000000}
\pgfsetfillcolor{dialinecolor}
% was here!!!
\definecolor{dialinecolor}{rgb}{0.000000, 0.000000, 0.000000}
\pgfsetstrokecolor{dialinecolor}
\draw (11.050000\du,9.600000\du)--(24.100000\du,8.350000\du);
}
\pgfsetlinewidth{0.100000\du}
\pgfsetdash{}{0pt}
\pgfsetdash{}{0pt}
\pgfsetbuttcap
{
\definecolor{dialinecolor}{rgb}{0.000000, 0.000000, 0.000000}
\pgfsetfillcolor{dialinecolor}
% was here!!!
\definecolor{dialinecolor}{rgb}{0.000000, 0.000000, 0.000000}
\pgfsetstrokecolor{dialinecolor}
\draw (11.009540\du,17.609540\du)--(24.059540\du,16.359540\du);
}
\pgfsetlinewidth{0.100000\du}
\pgfsetdash{{\pgflinewidth}{0.200000\du}}{0cm}
\pgfsetdash{{\pgflinewidth}{0.200000\du}}{0cm}
\pgfsetbuttcap
{
\definecolor{dialinecolor}{rgb}{0.000000, 0.000000, 0.000000}
\pgfsetfillcolor{dialinecolor}
% was here!!!
\definecolor{dialinecolor}{rgb}{0.000000, 0.000000, 0.000000}
\pgfsetstrokecolor{dialinecolor}
\draw (11.000000\du,17.600000\du)--(11.050000\du,9.600000\du);
}
\pgfsetlinewidth{0.100000\du}
\pgfsetdash{{\pgflinewidth}{0.200000\du}}{0cm}
\pgfsetdash{{\pgflinewidth}{0.200000\du}}{0cm}
\pgfsetbuttcap
{
\definecolor{dialinecolor}{rgb}{0.000000, 0.000000, 0.000000}
\pgfsetfillcolor{dialinecolor}
% was here!!!
\definecolor{dialinecolor}{rgb}{0.000000, 0.000000, 0.000000}
\pgfsetstrokecolor{dialinecolor}
\draw (24.000000\du,16.300000\du)--(24.050000\du,8.450000\du);
}
% setfont left to latex
\definecolor{dialinecolor}{rgb}{0.000000, 0.000000, 0.000000}
\pgfsetstrokecolor{dialinecolor}
\node[anchor=west] at (10.650000\du,9.150000\du){$i1$};
% setfont left to latex
\definecolor{dialinecolor}{rgb}{0.000000, 0.000000, 0.000000}
\pgfsetstrokecolor{dialinecolor}
\node[anchor=west] at (24.105000\du,17.452500\du){$i4$};
% setfont left to latex
\definecolor{dialinecolor}{rgb}{0.000000, 0.000000, 0.000000}
\pgfsetstrokecolor{dialinecolor}
\node[anchor=west] at (23.960000\du,7.807500\du){$j2$};
% setfont left to latex
\definecolor{dialinecolor}{rgb}{0.000000, 0.000000, 0.000000}
\pgfsetstrokecolor{dialinecolor}
\node[anchor=west] at (10.165000\du,17.412500\du){$j3$};
% setfont left to latex
\definecolor{dialinecolor}{rgb}{0.000000, 0.000000, 0.000000}
\pgfsetstrokecolor{dialinecolor}
% setfont left to latex
\definecolor{dialinecolor}{rgb}{0.000000, 0.000000, 0.000000}
\pgfsetstrokecolor{dialinecolor}
% setfont left to latex
\definecolor{dialinecolor}{rgb}{0.000000, 0.000000, 0.000000}
\pgfsetstrokecolor{dialinecolor}
% setfont left to latex
\definecolor{dialinecolor}{rgb}{0.000000, 0.000000, 0.000000}
\pgfsetstrokecolor{dialinecolor}
\pgfsetlinewidth{0.100000\du}
\pgfsetdash{}{0pt}
\pgfsetdash{}{0pt}
\pgfsetmiterjoin
\pgfsetbuttcap
{
\definecolor{dialinecolor}{rgb}{0.000000, 0.000000, 0.000000}
\pgfsetfillcolor{dialinecolor}
% was here!!!
\definecolor{dialinecolor}{rgb}{0.000000, 0.000000, 0.000000}
\pgfsetstrokecolor{dialinecolor}
\pgfpathmoveto{\pgfpoint{16.600000\du}{4.250000\du}}
\pgfpathcurveto{\pgfpoint{24.600000\du}{10.100000\du}}{\pgfpoint{7.250000\du}{15.050000\du}}{\pgfpoint{19.000000\du}{21.100000\du}}
\pgfusepath{stroke}
}
% setfont left to latex
\definecolor{dialinecolor}{rgb}{0.000000, 0.000000, 0.000000}
\pgfsetstrokecolor{dialinecolor}
\node[anchor=west] at (15.850000\du,5.100000\du){$S$};
% setfont left to latex
\definecolor{dialinecolor}{rgb}{0.000000, 0.000000, 0.000000}
\pgfsetstrokecolor{dialinecolor}
\node[anchor=west] at (17.855000\du,4.802500\du){$S^c$};
\end{tikzpicture}\caption{Illustration of $G_{k,p}$}
\end{figure}

Let $S=V^{1}\cup V^{3}$. Then, we have $\mathbb{E}\left(\cut{}S\right)=2\mathbb{E}(\sum_{i,j}H_{ij})=2pk^{2}$
and $\text{Var}\left(\cut{}S\right)=4\text{Var}(\sum_{i,j}H_{ij})=4p(1-p)k^{2}.$
Consider any deterministic algorithm that calls the oracle less than
$k^{2}/2$ times. Let $C$ be the estimate of $\cut{}S$ given by
the algorithm. Since each edge is only affected by one random variable
$H_{ij}$, only at most $k^{2}/2$ values of $H_{ij}$ are revealed.
Let $H$ be the set of known random variables $H_{ij}$. Then, we
have $\left|H\right|\leq k^{2}/2$. Therefore, the cut value $\cut{}S$
given $H$ follows the binomial distribution $2B(p,k^{2}-\left|H\right|)$
plus the constant $2\sum_{ij\in H}H_{ij}$. 

Since $p(k^{2}-\left|H\right|)\geq pk^{2}/2\geq50$, the result follows
from Lemma \ref{lem:binomial}.
\end{proof}
The following theorem shows that even the graph is quite sparse, it
is not possible to improve our probabilistic spectral sparsification
algorithm by too much. Instead of proving lower bound for the spectral
sparsification, we show the lower bound for the cut sparsification
which satisfies (\ref{eq:lower_bound}) and (\ref{eq:upper_bound})
for $u\in\{0,1\}^{V}$ only. 
\begin{thm}
\label{thm:lowerbound}For any $\varepsilon>0$ and $\delta>0$, it
takes $\Omega\left(\frac{n}{\varepsilon^{2}}+\frac{n}{\delta}\right)$
queries in the general graph model to construct a probabilistic $(\varepsilon,\delta)$
cut sparsifier for graphs with $n$ vertices and $\Omega\left(\frac{n}{\varepsilon^{2}}+\frac{n}{\delta}\right)$
edges.\end{thm}
\begin{proof}
We divide the proof into two cases, $\delta<\varepsilon^{2}$ and
$\delta\geq\varepsilon^{2}$. In both cases, we construct a family
of random graphs and shows that any deterministic algorithm takes
$\Omega\left(\frac{n}{\varepsilon^{2}}+\frac{n}{\delta}\right)$ queries
to estimate the cut value of a certain cut within the precision required.

For the first case $\delta<\varepsilon^{2}$, let $G$ be the disjoint
union of $\delta n$ independent copies of $G_{10\delta^{-1},\delta^{2}}$
defined in Lemma \ref{lem:Gkp}. Let $G_{i}$ be each copy and $S_{i}$
be each corresponding cut defined in Lemma \ref{lem:Gkp}. Note that
$G$ has $\Theta(n)$ vertices and $\Theta(\frac{n}{\delta})$ edges. 

Let us consider any deterministic algorithm which calls the oracle
less than $\frac{n}{4\delta}$ times. At least $\frac{\delta n}{2}$
copies of $G_{i}$, the algorithm calls the oracle less than $\frac{\delta^{-2}}{2}$
times for these $G_{i}$. Hence, Lemma \ref{lem:Gkp} shows with probability
$0.01$, the estimate value deviates from the cut value for more than
$1$. For those $S_{i}$, the estimate value is either larger than
the cut value by $1$ or is smaller than the cut value by $1$. Without
loss of generality, we assume the first case happens more. And let
$\mathcal{S}$ be the set of those $S_{i}$ in the first case. Then,
we have $|\mathcal{S}|=\Omega(\delta n)$ with high probability. Let
$A=\bigcup_{S\in\mathcal{S}}S$. Then, the estimate of $\cut{}A$
is larger than the true value by $\Omega(\delta n)$. Also, note that
$\cut{}A=O(\delta n)$.

It shows that any deterministic algorithm takes at least $\Omega(\frac{n}{\delta})$
queries to construct a probabilistic $(O(1),\delta)$ cut sparsifier
for graphs with $n$ vertices and $\Omega\left(\frac{n}{\delta}\right)$
edges.

For the second case $\delta\geq\varepsilon^{2}$, let $G$ be the
disjoint union of $\varepsilon^{2}n$ independent copies of $G_{10\varepsilon^{-2},\varepsilon^{2}}$.
By similar argument, we can show that any deterministic algorithm
takes at least $\Omega(\frac{n}{\varepsilon^{2}})$ queries to construct
a $(\varepsilon,O(1))$ cut sparsifier for graphs with $\Theta(n)$
vertices and $\Omega\left(\frac{n}{\varepsilon^{2}}\right)$ edges.

Combining both cases, the result follows from the Yao's principle.
\end{proof}
Similar lower bounds can be established for various problems. We use
the sparest cut problem as an example to show that our approach can
be used to give almost optimal results. 
\begin{thm}
\label{thm:lowerbound-1}For any $O(1)>\varepsilon>\frac{1}{n}$,
it takes $\Omega\left(\frac{n}{\varepsilon}\right)$ queries in the
general graph model to distinguish between a disconnected graph and
a graph with $\min_{|U|<\frac{n}{2}}\cut{}U/|U|=\Theta\left(\varepsilon\right).$\end{thm}
\begin{proof}
Let $G^{\varepsilon}=G_{10n,\varepsilon n^{-1}}$ defined in Lemma
\ref{lem:Gkp}. Put a complete graph inside $V^{1}$, $V^{2}$, $V^{3}$,
$V^{4}$ regions of $G^{\varepsilon}$ as defined in Lemma \ref{lem:Gkp}.
With high probability, we have $\min_{|U|<\frac{n}{2}}\cut{}U/|U|=\Theta\left(\varepsilon\right).$

Since $G^{\varepsilon}$ is a regular graph with same degree for all
$\varepsilon$, the degree oracle does not provide any information.
To distinguish between $G^{\varepsilon}$ and $G^{0}$, the algorithm
need to call the edge oracle until it found an edge from $V^{1}\cup V^{3}$
to $V^{2}\cup V^{4}$. Since the probability of finding such edge
is $O(\varepsilon n^{-1})$, it takes $\Omega\left(\frac{n}{\varepsilon}\right)$
queries to distinguish between $G^{\varepsilon}$ and $G^{0}$.
\end{proof}

\section{Acknowledgments}

We thank Ronitt Rubinfeld for many helpful conversations. This work
was partially supported by NSF awards 0843915 and 1111109, and Hong
Kong RGC grant 2150701.

\bibliographystyle{plain}
\bibliography{paper}

\begin{thebibliography}{10}

\bibitem{arora2009expander}
Sanjeev Arora, Satish Rao, and Umesh Vazirani.
\newblock Expander flows, geometric embeddings and graph partitioning.
\newblock {\em Journal of the ACM (JACM)}, 56(2):5, 2009.

\bibitem{benczur1996approximating}
Andr{\'a}s~A Bencz{\'u}r and David~R Karger.
\newblock Approximating st minimum cuts in {\~o} (n 2) time.
\newblock In {\em Proceedings of the twenty-eighth annual ACM symposium on
  Theory of computing}, pages 47--55. ACM, 1996.

\bibitem{chan2005image}
Tony Chan and Jianhong Shen.
\newblock {\em Image processing and analysis: variational, PDE, wavelet, and
  stochastic methods}.
\newblock Siam, 2005.

\bibitem{doerr2011analyzing}
Benjamin Doerr.
\newblock Analyzing randomized search heuristics: Tools from probability
  theory.
\newblock {\em Theory of randomized search heuristics}, 1:1--20, 2011.

\bibitem{frieze1999quick}
Alan Frieze and Ravi Kannan.
\newblock Quick approximation to matrices and applications.
\newblock {\em Combinatorica}, 19(2):175--220, 1999.

\bibitem{frieze2000edge}
Alan~M Frieze.
\newblock Edge-disjoint paths in expander graphs.
\newblock In {\em Proceedings of the eleventh annual ACM-SIAM symposium on
  Discrete algorithms}, pages 717--725. Society for Industrial and Applied
  Mathematics, 2000.

\bibitem{goel2010graph}
Ashish Goel, Michael Kapralov, and Sanjeev Khanna.
\newblock Graph sparsification via refinement sampling.
\newblock {\em arXiv preprint arXiv:1004.4915}, 2010.

\bibitem{goel2012single}
Ashish Goel, Michael Kapralov, and Ian Post.
\newblock Single pass sparsification in the streaming model with edge
  deletions.
\newblock {\em arXiv preprint arXiv:1203.4900}, 2012.

\bibitem{karger1997using}
David~R Karger.
\newblock Using random sampling to find maximum flows in uncapacitated
  undirected graphs.
\newblock In {\em Proceedings of the twenty-ninth annual ACM symposium on
  Theory of computing}, pages 240--249. ACM, 1997.

\bibitem{karger1998better}
David~R Karger.
\newblock Better random sampling algorithms for flows in undirected graphs.
\newblock In {\em Proceedings of the ninth annual ACM-SIAM symposium on
  Discrete algorithms}, pages 490--499. Society for Industrial and Applied
  Mathematics, 1998.

\bibitem{karger1998finding}
David~R Karger and Matthew~S Levine.
\newblock Finding maximum flows in undirected graphs seems easier than
  bipartite matching.
\newblock In {\em Proceedings of the thirtieth annual ACM symposium on Theory
  of computing}, pages 69--78. ACM, 1998.

\bibitem{kelner2011spectral}
Jonathan Kelner, Alex Levin, et~al.
\newblock Spectral sparsification in the semi-streaming setting.
\newblock {\em Leibniz International Proceedings in Informatics (LIPIcs)
  series}, 9:440--451, 2011.

\bibitem{kelner2013almost}
Jonathan~A Kelner, Lorenzo Orecchia, Yin~Tat Lee, and Aaron Sidford.
\newblock An almost-linear-time algorithm for approximate max flow in
  undirected graphs, and its multicommodity generalizations.
\newblock {\em arXiv preprint arXiv:1304.2338}, 2013.

\bibitem{kelner2013simple}
Jonathan~A Kelner, Lorenzo Orecchia, Aaron Sidford, and Zeyuan~Allen Zhu.
\newblock A simple, combinatorial algorithm for solving sdd systems in
  nearly-linear time.
\newblock In {\em Proceedings of the 45th annual ACM symposium on Symposium on
  theory of computing}, pages 911--920. ACM, 2013.

\bibitem{koutis2012faster}
Ioannis Koutis, Alex Levin, and Richard Peng.
\newblock Faster spectral sparsification and numerical algorithms for sdd
  matrices.
\newblock {\em arXiv preprint arXiv:1209.5821}, 2012.

\bibitem{koutis2011nearly}
Ioannis Koutis, Gary~L Miller, and Richard Peng.
\newblock A nearly-m log n time solver for sdd linear systems.
\newblock In {\em Foundations of Computer Science (FOCS), 2011 IEEE 52nd Annual
  Symposium on}, pages 590--598. IEEE, 2011.

\bibitem{lee2013new}
Yin~Tat Lee, Satish Rao, and Nikhil Srivastava.
\newblock A new approach to computing maximum flows using electrical flows.
\newblock In {\em STOC}, pages 755--764, 2013.

\bibitem{lee2013efficient}
Yin~Tat Lee and Aaron Sidford.
\newblock Efficient accelerated coordinate descent methods and faster
  algorithms for solving linear systems.
\newblock {\em arXiv preprint arXiv:1305.1922}, 2013.

\bibitem{levy2006laplace}
Bruno L{\'e}vy.
\newblock Laplace-beltrami eigenfunctions towards an algorithm that.
\newblock In {\em Shape Modeling and Applications, 2006. SMI 2006. IEEE
  International Conference on}, pages 13--13. IEEE, 2006.

\bibitem{lovasz1993random}
L{\'a}szl{\'o} Lov{\'a}sz.
\newblock Random walks on graphs: A survey.
\newblock {\em Combinatorics, Paul erdos is eighty}, 2(1):1--46, 1993.

\bibitem{lubotzky1988ramanujan}
Alexander Lubotzky, Ralph Phillips, and Peter Sarnak.
\newblock Ramanujan graphs.
\newblock {\em Combinatorica}, 8(3):261--277, 1988.

\bibitem{madry2010fast}
Aleksander Madry.
\newblock Fast approximation algorithms for cut-based problems in undirected
  graphs.
\newblock In {\em Foundations of Computer Science (FOCS), 2010 51st Annual IEEE
  Symposium on}, pages 245--254. IEEE, 2010.

\bibitem{marko2006distance}
Sharon Marko and Dana Ron.
\newblock Distance approximation in bounded-degree and general sparse graphs.
\newblock In {\em Approximation, Randomization, and Combinatorial Optimization.
  Algorithms and Techniques}, pages 475--486. Springer, 2006.

\bibitem{nguyen2008constant}
Huy~N Nguyen and Krzysztof Onak.
\newblock Constant-time approximation algorithms via local improvements.
\newblock In {\em Foundations of Computer Science, 2008. FOCS'08. IEEE 49th
  Annual IEEE Symposium on}, pages 327--336. IEEE, 2008.

\bibitem{onak2012near}
Krzysztof Onak, Dana Ron, Michal Rosen, and Ronitt Rubinfeld.
\newblock A near-optimal sublinear-time algorithm for approximating the minimum
  vertex cover size.
\newblock In {\em Proceedings of the Twenty-Third Annual ACM-SIAM Symposium on
  Discrete Algorithms}, pages 1123--1131. SIAM, 2012.

\bibitem{parnas2007approximating}
Michal Parnas and Dana Ron.
\newblock Approximating the minimum vertex cover in sublinear time and a
  connection to distributed algorithms.
\newblock {\em Theoretical Computer Science}, 381(1):183--196, 2007.

\bibitem{sherman2009breaking}
Jonah Sherman.
\newblock Breaking the multicommodity flow barrier for o (vlog
  n)-approximations to sparsest cut.
\newblock In {\em Foundations of Computer Science, 2009. FOCS'09. 50th Annual
  IEEE Symposium on}, pages 363--372. IEEE, 2009.

\bibitem{sherman2013nearly}
Jonah Sherman.
\newblock Nearly maximum flows in nearly linear time.
\newblock {\em arXiv preprint arXiv:1304.2077}, 2013.

\bibitem{spielman2011graph}
Daniel~A Spielman and Nikhil Srivastava.
\newblock Graph sparsification by effective resistances.
\newblock {\em SIAM Journal on Computing}, 40(6):1913--1926, 2011.

\bibitem{spielman2011spectral}
Daniel~A Spielman and Shang-Hua Teng.
\newblock Spectral sparsification of graphs.
\newblock {\em SIAM Journal on Computing}, 40(4):981--1025, 2011.

\bibitem{yoshida2009improved}
Yuichi Yoshida, Masaki Yamamoto, and Hiro Ito.
\newblock An improved constant-time approximation algorithm for maximum.
\newblock In {\em Proceedings of the 41st annual ACM symposium on Theory of
  computing}, pages 225--234. ACM, 2009.

\end{thebibliography}

\section*{Appendix}
\begin{lem}
\label{lem:binomial}Let $0\leq p\leq1/4$ and $n$ be an integer
such that $pn\geq36$. Let $X\sim B\left(p,n\right)$. Then, for any
$\theta$, we have 
\[
\mathbb{P}\left(\left|X-\theta\right|\geq\frac{1}{2}\sqrt{pn}\right)\geq0.01.
\]
\end{lem}
\begin{proof}
Note that for any $\theta$, we have 
\[
\mathbb{P}\left(\left|X-\theta\right|\geq\frac{1}{2}\sqrt{pn}\right)\geq\mathbb{P}\left(\left|X-pn\right|\geq\frac{1}{2}\sqrt{pn}\right)
\]
because of the shape of the binomial distribution. Hence, it suffices
to prove the bound for $\mathbb{P}\left(\left|X-pn\right|\geq\frac{1}{2}\sqrt{pn}\right)$. 

Using Chernoff bound, for any $k\geq6$, we have
\begin{eqnarray*}
\mathbb{P}\left(\left|X-pn\right|\geq k\sqrt{pn}\right) & \leq & 2\exp\left(-\frac{k^{2}}{2+\frac{k}{\sqrt{pn}}}\right)\\
 & \leq & 2\exp\left(-2k\right).
\end{eqnarray*}
Hence, for $k\geq6$, we have
\begin{eqnarray*}
\int_{|x-pn|\geq k\sqrt{pn}}(x-pn)^{2}dP(x) & = & 2k^{2}pn\mbox{\ensuremath{\mathbb{P}}}\left(X\geq pn+k\sqrt{pn}\right)\\
 &  & +4\int_{x\geq pn+k\sqrt{pn}}(x-pn)\mbox{\ensuremath{\mathbb{P}}}\left(X\geq x\right)dx\\
 & \leq & 2k^{2}pn\exp\left(-2k\right)+4\int_{k\sqrt{pn}}^{\infty}x\exp\left(-2\frac{x}{\sqrt{pn}}\right)dx\\
 & = & (2k^{2}+2k+1)pn\exp\left(-2k\right).
\end{eqnarray*}
Put $k=6$, we have $\int_{|x-pn|\geq6\sqrt{pn}}(x-pn)^{2}dP\leq0.01pn$.
Since the $\text{Var}(X)=\int(x-pn)^{2}dP=p(1-p)n\geq\frac{3}{4}pn$,
we have
\begin{eqnarray*}
\int_{|x-pn|<6\sqrt{pn}}(x-pn)^{2}dP & \geq & 0.74pn.
\end{eqnarray*}
Let $U=P\left(\left|X-pn\right|\geq\frac{1}{2}\sqrt{pn}\right)$,
then we have 
\begin{eqnarray*}
36Upn+(1-U)\frac{pn}{4} & \geq & \int_{|x-pn|<6\sqrt{pn}}(x-pn)^{2}dP\\
 & \geq & 0.74pn.
\end{eqnarray*}
Hence, we have $U\geq0.01$.\end{proof}

\end{document}